\documentclass[authoryear]{elsarticle}

\makeatletter
\def\ps@pprintTitle{%
 \let\@oddhead\@empty
 \let\@evenhead\@empty
 \def\@oddfoot{}%
 \let\@evenfoot\@oddfoot}
\makeatother

\usepackage[normalem]{ulem}
\usepackage{amsmath,amsthm}
\usepackage[a4paper,width=15cm,bindingoffset=0cm]{geometry}
\usepackage[dvipsnames]{xcolor}
\usepackage[colorlinks]{hyperref}
\usepackage{amsfonts}
\usepackage{bbm}

\newcommand\myshade{85}
\colorlet{mylinkcolor}{OrangeRed}
\colorlet{mycitecolor}{YellowOrange}
\colorlet{myurlcolor}{Aquamarine}

\hypersetup{
	pdftitle={Network Valuation in Financial Systems},
	pdfauthor={Paolo Barucca, Marco Bardoscia, Fabio Caccioli, Marco D'Errico, Gabriele Visentin, Guido Caldarelli, Stefano Battiston},
	linkcolor=mylinkcolor!\myshade!black,
	citecolor=mycitecolor!\myshade!black,
	urlcolor=myurlcolor!\myshade!black
}

\DeclareSymbolFont{largesymbolsA}{U}{txexa}{m}{n}
\DeclareMathSymbol{\varprod}{\mathop}{largesymbolsA}{16}

\newtheorem{defn}{Definition}[section]
\newtheorem{theorem}{Theorem}[section]
\newtheorem{proposition}{Proposition}[section]

\begin{document}
\begin{frontmatter}

\title{Network Valuation in Financial Systems}

\author[a4]{Paolo Barucca}
\author[a3]{Marco Bardoscia \fnref{fn1}}
\author[a4,a5,a10]{Fabio Caccioli}
\author[a8]{Marco D'Errico \fnref{fn2}}
\author[a11]{Gabriele Visentin}
\author[a6,a7,a9,a2]{Guido Caldarelli}
\author[a1]{Stefano Battiston}

\address[a4]{Department of Computer Science, University College London, WC1E 6BT London, United Kingdom}
\address[a3]{Bank of England, EC2R 8AH London, United Kingdom}
\address[a5]{Systemic Risk Centre, London School of Economics, WC2A 2AE London, United Kingdom}
\address[a10]{London Mathematical Laboratory, London WC 8RH, UK}
\address[a8]{European Systemic Risk Board Secretariat, European Central Bank, 60640 Frankfurt am Main, Germany}
\address[a11]{Institute of Mathematics, University of Zurich, 8057 Zurich, Switzerland}
\address[a6]{IMT School for Advance Studies Lucca, 55100 Lucca, Italy}
\address[a7]{Istituto dei Sistemi Complessi (CNR) UoS Sapienza, Dipartimento di Fisica, Sapienza Universit\`{a} di Roma, 00185 Rome, Italy}
\address[a9]{European Centre for Living Technology, Universit\`{a} di Venezia Ca'\ Foscari, 30124 Venice, Italy}
\address[a2]{London Institute for Mathematical Sciences, W1K 2XF London, United Kingdom}
\address[a1]{Department of Banking and Finance, University of Z\"{u}rich, 8050 Z\"{u}rich, Switzerland}

\fntext[fn1]{Any views expressed are solely those of the author(s) and so cannot be taken to represent those of Bank of England or to state Bank of England policy.}
\fntext[fn2]{The paper should not be reported as representing the views of the European Central Bank (ECB) or the European Systemic Risk Board (ESRB), or any of their member institutions. The views are those of the authors.}

\begin{abstract}
We introduce a general model for the balance-sheet consistent valuation of interbank claims within an interconnected financial system. Our model represents an extension of clearing models of interdependent liabilities to account for the presence of uncertainty on banks' external assets. At the same time, it also provides a natural extension of classic structural credit risk models to the case of an interconnected system. 
We characterize the existence and uniqueness of a valuation that maximises individual and total equity values for all banks. We apply our model to the assessment of systemic risk, and in particular for the case of stress-testing.
Further, we provide a fixed-point algorithm to carry out the network valuation and the conditions for its convergence. 


\end{abstract}

\end{frontmatter}


\section{Introduction}

\noindent Uncertainty and interdependence are two fundamental features of financial systems. 
While uncertainty\footnote{The term uncertainty is used here in its generic sense of lack of certainty over future values, regardless of whether the probability distribution of such values is known or not.} over the future value of assets is traditionally very central in the financial literature \citep{merton1974pricing,black1976valuing,leland1996optimal,duffie1999modeling,lando2009credit}, interdependence of financial claims' values, especially in the increasingly  interconnected financial system, has been investigated only more recently \citep{rochet1996Interbank,freixas2000systemicrisk,allen2001financial,gray2010measuring}, taking center stage mostly after the 2008 financial crisis \citep{stiglitz2010risk,allen2013systemic,lewandowska2015otc,elliott2014financial,acemoglu2015systemic,glasserman2015likely,amini2016resilience,battiston2016price,bardoscia2017pathways}. 

The combined presence of both uncertainty and interdependence is the most relevant situation in practice, and yet the valuation of assets in this case remains an open problem in financial economics. 
Indeed, in their daily business, financial institutions need to assign an economic value to the claims they hold on their counterparties. 
For example, whether the counterparties will actually pay their obligations at maturity will depend on the counterparties' financial situation at maturity, which is known by other institutions only with some uncertainty, for instance because of exogenous changes in the value of their loans to the real sector. 
The obligations of, say, institution A are assets for A's creditors.
If more information becomes available on A, its creditors mark to market those assets, incorporating such information in the valuation of their own balance sheets.
As more information is now available on A's creditors, the creditors of A's creditors will in turn update the valuation of their own assets, and so on.

Neglecting either interconnectedness or uncertainty may lead to misestimation of systemic risk. 
On the one hand, not accounting for interconnectedness amounts to 
considering only direct exposures, ignoring potential indirect exposures to counteparties of counterparties, and so on. 
This, in turn, may be reflected in the overvaluation of counterparties' obligations and in incentives for excessive risk taking. 
On the other hand, neglecting uncertainty implies to consider that losses cannot materialize before the maturity. 
In contrast, according to the Bank for International Settlements (BIS), the largest part of losses suffered by financial institutions during the financial crisis was not due to actual counterparties' defaults, but to the mark-to-market re-evaluation of obligations following the deterioration of counterparties' creditworthiness.\footnote{The Basel Committee on Banking Supervision states that ``roughly two-thirds of losses attributed to counterparty credit risk were due to Credit Valuation Adjustment (CVA) losses and only about one-third were due to actual defaults'', see \cite{basel2011basel}.} 
The case of AIG in 2008 illustrates the distinction between the propagation of realized losses and ``mark-to-market'' losses, whereby the deterioration of institutions creditworthiness can spread through the network well before the maturity of the contracts \citep{glasserman2016contagion}.

Our main contribution is to introduce a framework to perform asset valuation taking into account both the interdependence of balance sheets connected by a network of so-called interbank assets and the uncertainty on future values of banks' external (i.e.\ non-interbank) assets. 
The framework abstracts away from the details of how the valuation is performed. 
Such details are encapsulated into \emph{valuation functions}. Once valuation functions are chosen, the framework reduces to a specific model.  
The framework includes two families of models.
If the valuation is performed  \emph{strictly before} the maturity of interbank claims, we have a proper model for ex-ante valuation.
If the valuation is performed \emph{at} the maturity, the model describes the clearing of interbank claims.
Here we provide general results that apply to all models compatible with the framework. 
In particular we cast the problem into a set of fixed-point equations for the valuation of institutions' equities.
We prove the existence of a solution that is optimal for all institutions, and we provide a simple algorithm to compute it with arbitrary precision.

Our second contribution is to show that, by suitably choosing the valuation functions, our framework recovers several models of clearing \citep{eisenberg2001systemic,rogers2013failure} and ex-ante valuation \citep{furfine2003quantifying,bardoscia2015debtrank} previously introduced in the literature.

Finally, through our framework we establish a specific connection between clearing models and ex-ante valuation models.
In particular, we start from the model introduced in \cite{eisenberg2001systemic} (EN) and we show that, by averaging its equations over the ex-ante uncertainty, we obtain an ex-ante valuation model in the sense implied by our framework. 
From this point of view the ex-ante valuation model can be seen as a forward-looking extension of the corresponding clearing model.

Our work is related to several strands of literature. 
First, adjusting the value of a contract between two counterparties to account for the risk that they might default is typically referred to in the literature as Credit Valuation Adjustment (CVA), see e.g.\ \cite{sorensen1994pricing,bielecki2013credit}. In its most basic form, CVA is computed by one institution as the risk-neutral expectation of the losses that it would incur if its counterparty were to default. 
In practice, in order to compute such expectation, one typically assumes a specific exogenous stochastic process for the probability of default of the counterparty. 
Moreover, as pointed out in \cite{banerjee2018pricing}, CVA only captures adjustments due to potential defaults of direct counterparties, but not of indirect counterparties.
In contrast, in our framework probabilities of default are computed endogenously (see section \ref{sec:ex-ante}) and account for the creditworthiness of both direct and indirect counterparties.\footnote{However, we stress that, since our valuations emerge as solution of a system of fixed-point equations, defaulting institutions do not default in a specific sequence, but default all at the same time. As a consequence, our approach is not well-suited to capture bilateral CVA \citep{gregory2009being,brigo2010bilateral}, which also accounts for the sequence in which defaults might occur.}

Second, because our framework includes clearing models, it is naturally related to the literature on the clearing of payments between institutions with mutual obligations.
The most widely used model for clearing payments is the EN model, which has been extended to the case of non-zero bankruptcy costs by \cite{rogers2013failure} (RV), cross-holdings of equities \citep{suzuki2002valuing}, an arbitrary seniority structure of claims \citep{fischer2014no-arbitrage}, and time-varying balance sheets \citep{banerjee2018dynamic}.
Both EN and RV can be explicitly recovered from our framework. Early attempts to establish network valuation models can be traced back to the work of \cite{furfine2003quantifying}, who introduces a model for cascades of defaults where the value of an interbank asset is equal to its face value as long as the debtor has not defaulted, and it is equal to zero otherwise. 
Also the Furfine model can be explicitly recovered from our framework.
In models of cascades of defaults, the deterioration of credit worthiness of an institution does not have any consequence for its creditors as long as that institution does not default. 
To overcome this limitation, a simple mechanism, the so-called DebtRank, has been introduced in \cite{battiston2012debtrank,bardoscia2015debtrank,battiston2016leveraging}, based on the assumption that relative shocks to equities of debtors are linearly transmitted to interbank assets of creditors. 
Further extensions generalize the model to non-linear transmission mechanisms \citep{bardoscia2015distress,bardoscia2017pathways}.
DebtRank can also be explicitly recovered from our framework.
All the aforementioned models are models of direct contagion, where shocks are propagated via direct exposures. We note that propagation of shocks in absence of defaults can occur because of indirect contagion, for instance in the case of overlapping portfolios and fire-sales \citep{cifuentes2005liquidity,caccioli2014stability,amini2016uniqueness,cont2017fire,feinstein2017financial}.

Recently, \cite{veraart2018distress} has followed an approach similar to ours, but in which the valuation mechanism is introduced as an assumption.
In contrast, here, and similarly to \cite{elsinger2005using,elsinger2006risk,fischer2014no-arbitrage}, the valuation mechanism is derived from a clearing mechanism. 
In \cite{elsinger2005using,elsinger2006risk,fischer2014no-arbitrage} the valuation can be performed only by an agent who has complete knowledge of all institutions' interbank assets.
Similar set-ups are presented in \cite{collin2004general} and \cite{cossin2007credit}.
However, in practice, interbank assets encode financial exposures between institutions and are therefore confidential. 
Often, even regulators have only a partial view of the detailed structure of interbank assets.
Indeed, a separate strand of literature is focused specifically on trying to ``reconstruct'' interbank assets from publicly available information \citep{anand2015filling,gandy2016bayesian,cimini2015systemic,squartini2017network,squartini2018reconstruction} or to assess the impact of their misestimation \citep{feinstein2018sensitivity}.

In line with the above consideration, in our framework institutions are assumed to need only knowledge of their own interbank assets, while the valuation is performed collectively by all institutions.
The basic idea is that each institution performs a valuation of its interbank assets, which is reflected in the value of its equity.
Counterparties that hold claims towards this institution, in turn, update the valuation of their own interbank assets, which is eventually reflected in the value of their own equities.
Our paper provides a characterization of the solutions to this valuation process.

\section{Framework} \label{sec:framework}
\noindent We consider a financial system consisting of $n$ institutions (for brevity ``banks'' hereafter) engaging in credit contracts with each other. 
These contracts mature at time $T$, while banks perform a valuation of their assets at time $t \leq T$.
At time $t$, we denote with $L_{ij}(t)$ the book value of the liability of bank $i$ towards bank $j$ and with $A_{ji}(t)$ the book value of the corresponding asset of bank $j$, with $A_{ji}(t) = L_{ij}(t)$ for consistency.
We refer to these quantities as interbank assets and liabilities.

We point out that, in general, book values of interbank assets and liabilities can depend on time, but they are deterministic and, crucially, do not incorporate any information about the creditworthiness of counterparties. 
In this respect, one could think of $A_{ji}(T)$ as the amount that $j$ is expected to recover from $i$ at maturity, if $i$ is not in default.
Similarly, if no additional contracts are stipulated between bank $i$ and $j$ in the period from $t$ to $T$, one could think of $A_{ji}(t)$ as the discounted value of $A_{ji}(T)$. 
Banks also have external, i.e.\ non-interbank, assets and liabilities. 
For example, external assets include loans to the real sector, while external liabilities include deposits.
At time $t$, we denote with $A_i^e(t)$ the book value of external assets of bank $i$ and with $L_i^e(t)$ the book value of its external liabilities.
External liabilities are deterministic, while external assets follow a stochastic process. 
Each bank observes its external assets at the valuation time $t$, but they are in general unknown at any time between $t$ (excluded) and the maturity $T$ (included, unless $t = T$).\footnote{Therefore, with a slight abuse of notation we denote with $A_i^e(t)$ the realisation at $t$ of the stochastic process of external assets of bank $i$, while for $s > t$ we denote with $A_i^e(s)$ the random variable corresponding to external assets of bank $i$ at time $s$.}
Finally, we denote with $M_i(t)$ the book value of the equity of bank $i$ at time $t$, i.e.\ the difference between its total assets and liabilities, taken at their book value:
\begin{equation} \label{eq:naive_equity}
M_i(t) = A_i^e(t) - L_i^e(t) + \sum_{j=1}^n A_{ij}(t) - \sum_{j=1}^n L_{ij}(t) \, .
\end{equation}

At time $t$, banks perform a valuation of their own interbank assets. 
For, say, bank $i$, the purpose of such valuation is to incorporate any information about the creditworthiness of $i$'s debtors into the value of $i$'s interbank assets.
The valuation will depend in principle on the information available to $i$.
Without loss of generality, we write the valuation of interbank assets in the following way:
\begin{equation} \label{eq:cptyvaluation}
A_{ij}(t) \mathbb{V}_{ij} ( \mathbf{O}_i(t) ) \, ,
\end{equation}
where $\mathbf{O}_i(t)$ is the information, i.e.\ the list of variables and parameters used by $i$ to perform the valuation of interbank assets at time $t$.
We call $\mathbb{V}_{ij}$ \emph{interbank valuation function}, which, at this stage, is simply the ratio between the valuation of the interbank asset and its book value.
Because the purpose of $i$'s interbank valuation function is to account for the creditworthiness of $i$'s debtors, we can expect that $\mathbf{O}_i(t)$ includes information on them. 
The precise list of variables and parameters part of $\mathbf{O}_i(t)$ will depend on the specific valuation model, but in all cases those are deterministic, precisely because they are observed by $i$ at time $t$.
Also, we note that in general the functional form of $\mathbb{V}_{ij}$ depends explicitly on both $i$ and $j$. 
This reflects the fact that different banks could use different models to perform the valuation of their interbank assets, and also that one bank could use different models to perform the valuation of different interbank assets.

Similarly, banks perform also a valuation of their external assets. 
In this case we write:
\begin{equation} \label{eq:val_ext}
A_i^e(t) \mathbb{V}_i^e ( \mathbf{O}_i^e(t) ) \, .
\end{equation}
Analogously, $\mathbf{O}_i^e(t)$ is the information, i.e.\ the list of variables and parameters used by $i$ to perform the valuation of external assets at time $t$ and we call $\mathbb{V}_i^e$, i.e.\ the ratio between the valuation of external assets and their book value, \emph{external valuation function}.
The external valuation function can be used, for example, to account for the loss implied by the fire sale of external assets.
Let us imagine that banks were to target a certain leverage ratio and that, whenever they were to deviate from their target, they would deleverage by selling external assets below market price (a similar mechanism is described in \cite{greenwood2015vulnerable,battiston2016leveraging,cont2017fire}).
The external valuation function would then be interpreted as the discount at which external assets were sold.
Such discount factor would depend both on bank $i$'s leverage and on other quantities, e.g.\ the market price and liquidity of its external assets, which would be part of $\mathbf{O}_i^e(t)$.
Also in this case the precise list of variables and parameters part of $\mathbf{O}_i^e(t)$ will depend on the specific valuation model, but analogously to $\mathbf{O}_i(t)$ those are deterministic.

Banks do not perform a valuation of their liabilities. The rationale of this assumption is that banks are not allowed to discount their liabilities based on their own creditworthiness or on the creditworthiness of their creditors.
This is consistent with expecting that creditors will try to recover the full value of their claims towards their debtors.
By replacing book values of assets with their valuations in \eqref{eq:naive_equity} we obtain the equity valuation (for brevity equity hereafter) for bank $i$:
\begin{equation} \label{eq:balance_sheet_noequity}
E_i(t) = A_i^e(t) \mathbb{V}_i^e ( \mathbf{O}_i^e(t) )- L_i^e(t) + \sum_{j=1}^n A_{ij}(t) \mathbb{V}_{ij} ( \mathbf{O}_i(t) ) - \sum_{j=1}^n L_{ij}(t) \, .
\end{equation}

We assume that, for all banks, the total book value of interbank assets and liabilities,\footnote{The total book value of interbank assets of bank $i$ at time $t$ is equal to $\sum_j A_{ij}(t)$, while total book value of interbank liabilities is equal to $\sum_j L_{ij}(t)$.} as well as the book value of external assets and liabilities at time $t$ are public information.
For example, at each quarter $t$, those values could be taken from banks' financial accounts.
However, individual values of interbank assets (and therefore of interbank liabilities) are not public, and are known only by the two participating counterparties. 
In our framework this is indeed the only piece of private information.
We assume that equities as well, which incorporate the asset valuations, are also part of the public information. 
This would happen, for example, if banks communicated their equity valuations to their counterparties.
In the following, we explicitly highlight the dependence of valuation functions on equities by writing them as $\mathbb{V}_i^e ( \mathbf{E}(t) | \mathbf{O}_i^{e \prime}(t) )$ and $\mathbb{V}_{ij} ( \mathbf{E}(t) | \mathbf{O}'_i(t) )$, where $\mathbf{O}_i^{e \prime}(t)$ and $\mathbf{O}'_i(t)$ are the lists of variables and parameters used by $i$ at time $t$ to perform the valuation of, respectively, external and interbank assets, \emph{in addition} to equities.
Formally, $\mathbf{O}_i(t) = \mathbf{O}'_i(t) \cup \{\mathbf{E}(t)\}$ and $\mathbf{O}_i^e(t) = \mathbf{O}_i^{e \prime}(t) \cup \{\mathbf{E}(t)\}$.
This allows us to re-write \eqref{eq:balance_sheet_noequity} as:
\begin{equation} \label{eq:balance_sheet}
\begin{split}
E_i(t) =& \, A_i^e(t) \mathbb{V}_i^e (  \mathbf{E}(t) | \mathbf{O}_i^{e \prime}(t)  )- L_i^e(t) \\
&+ \sum_{j=1}^n A_{ij}(t) \mathbb{V}_{ij} (  \mathbf{E}(t) | \mathbf{O}'_i(t) ) - \sum_{j=1}^n L_{ij}(t) \, .
\end{split}
\end{equation}
The third term on the right-hand side of \eqref{eq:balance_sheet} accounts for the valuation of $i$'s interbank assets towards its \emph{direct} counterparties.
However, by solving \eqref{eq:balance_sheet} jointly for all $i$ we  account not only for the effect of direct counterparties, but also for genuine network effects arising from counterparties of counterparties, and so on.
In fact, the analogous of \eqref{eq:balance_sheet} for $i$'s counterparties includes the valuation of \emph{their} interbank assets towards their own direct counterparties, which are indirect counterparties for $i$.
Hence, the joint solution of \eqref{eq:balance_sheet} for all $i$ effectively accounts for the valuations of interbank assets of all indirect counterparties.

When $t < T$, banks perform a proper ex-ante valuation.
External assets follow a stochastic process and therefore that their (book) value at maturity is unknown at time $t$.
In turn, this generates uncertainty on banks' solvency at maturity.
Intuitively, the valuation of interbank assets (via valuation functions) incorporates creditors' estimate of the likelihood that their debtors will be able to meet their obligations at maturity, given this source of uncertainty.
Eq.\ \eqref{eq:balance_sheet} is also valid when $t = T$, i.e.\ if banks perform their valuation \emph{at} maturity.
In this case, there is no uncertainty on the (book) value of external assets.
Nevertheless, creditors are still not fully certain about the value of their interbank assets, until they actually receive payments from their debtors. Whether their direct debtors are able to deliver such payments could depend on whether the debtors of their debtors are able to deliver payments. 
In this case, our framework is consistent with a setting in which the valuation happens when payments are due, but before they are delivered.
Interbank valuation functions then incorporate creditors' estimate of the likelihood that their debtors will deliver their payments.
In this sense, the joint solution of \eqref{eq:balance_sheet} amounts to clearing payments between banks.

In this section we will introduce a precise definition of valuation function that will allow us to prove general results that hold regardless of their specific functional form.
Those results, derived in section \ref{sec:main_results}, will not rely on any further assumption.
In section \ref{sec:ex-ante} we will derive a specific set of valuation functions and we will discuss its economic meaning in detail.
As anticipated, the information used by banks to perform the valuation depends on the specific valuation model. 
For example, for the interbank valuation functions derived in section \ref{sec:ex-ante}, $\mathbf{O}'_i(t)$ will include external assets of $i$'s counterparties at time $t$, their volatilities, and distance to maturity $T - t$.

In principle a valuation function can depend on the equities of all banks, but in almost all the examples that we will make their dependence on equities will be much simpler.
Because external assets, by definition, are independent of any specific counterparty, in most cases external valuation functions will depend only on the equity of the bank performing the valuation, i.e.\ \eqref{eq:val_ext} will read $\mathbb{V}_i^e (  E_i(t) | \mathbf{O}_i^{e \prime}(t) )$.
Analogously, because interbank valuation functions are meant to capture the credit risk of interbank assets, interbank valuation functions will depend only on the equity of the debtor, i.e.\ \eqref{eq:cptyvaluation} will read $\mathbb{V}_{ij} ( E_j(t) | \mathbf{O}'_i(t) )$.
Nevertheless, the results that we prove in section \ref{sec:main_results} apply to the more general case in which valuation functions depend on the equities of any subset of banks.
This motivates the following definition.

\begin{defn} \label{eq:val_fun}
Given an integer $q \leq n$, a function $\mathbbm{V}: \mathbb{R}^q \to [0, 1]$ is called a \emph{feasible valuation function} if and only if:
\begin{enumerate}
\item it is nondecreasing: $\mathbf{E} \leq \mathbf{E}^\prime$ $\Rightarrow$ $\mathbbm{V}(\mathbf{E}) \leq \mathbbm{V}(\mathbf{E}^\prime), \forall \mathbf{E}, \mathbf{E}^\prime \in \mathbb{R}^q$
\item it is continuous from above.
\end{enumerate}
\end{defn}

The first observation is that a feasible valuation function takes values between zero and one.
This, combined with \eqref{eq:cptyvaluation} and \eqref{eq:val_ext}, corresponds to assuming that both interbank and external assets cannot be valued at more than their book value.
This is obvious for interbank assets, which represent credit contracts, as the creditor cannot expect to recover more than the book value of the contract.
As already mentioned, the external valuation function will allow us to model bankruptcy costs due to fire sales. 
From this point of view this assumption is consistent with the fact that a bank cannot expect to profit from selling illiquid assets.

The second observation is that a feasible valuation function is non-decreasing.
For interbank assets this corresponds to assuming that credit contracts are non-speculative, in the sense that the valuation of interbank assets of one bank cannot increase because the equity of another bank, for example of one of its debtors, has decreased.\footnote{Such notion of non-speculative contracts is similar to the ones introduced in \cite{schuldenzucker2019default} and further extended in \cite{banerjee2019impact}.}
We have already mentioned that in all cases considered in sections \ref{sec:relation} and \ref{sec:ex-ante} interbank valuation functions depend only on the equity of the debtor. 
In those specific cases we have also that $\lim_{E_j(t) \to +\infty} \mathbb{V}_{ij} ( E_j(t) | \mathbf{O}'_i(t) ) = 1$.
This simply means that, when the equity of bank $j$ is very large, bank $i$ will deem bank $j$ so creditworthy that the corresponding interbank asset is taken at book value.
However, this property is not explicitly required by Definition \ref{eq:val_fun} and indeed it is not necessary for the results in section \ref{sec:main_results}.
In this case, one could interpret the expression $A_{ij}(t) \left[ 1 - \mathbb{V}_{ij} ( E_j(t) | \mathbf{O}'_i(t) ) \right]$ as CVA losses. In fact, this is the difference between the book value of $i$'s interbank assets towards $j$ and its valuation incorporating the information about $j$'s creditworthiness.  
Indeed, it is equal to zero when the interbank valuation function is equal to one, corresponding to the case in which no credit valuation adjustment needs to be applied.
For external assets this means that their value cannot be boosted by a decrease in the equity of the bank that holds them (or of any other bank).

The assumption of continuity from above is mainly technical and allows us to deal with the corner case of discontinuous valuation functions.
Indeed, most valuation functions that we will introduce are continuous (i.e.\ both from above and from below).
A discontinuity of the valuation function corresponds to a finite jump in the valuation of assets following an infinitesimal change in the value of equities.
For example, let us imagine the case in which the interbank valuation function captures the extremely simplified situation in which creditors take interbank assets at book value as long as their debtor has positive equity, while they value interbank asset as worthless otherwise.
What shall a creditor do when the equity of one of its debtors is exactly equal to zero, i.e.\ when assets and liabilities of that debtor are exactly equal?
Continuity from above implies that in this case the creditor should still take interbank assets at book value.

Since all valuation functions take values in the interval $[0, 1]$, all equities $E_i(t)$ are bounded both from below and from above as follows:
\begin{equation}\label{eq:consistency}
m_i(t) \equiv - L_i^e(t) - \sum_{j=1}^n L_{ij}(t) \leq E_i(t) \leq M_i(t) \, .
\end{equation}
By introducing the following map:
\begin{subequations}
\begin{equation}
\Phi \; : \; \varprod_{i=1}^n [m_i(t),M_i(t)] \to \varprod_{i=1}^n [m_i(t),M_i(t)]
\end{equation}
\begin{equation}
\Phi = (\Phi_1, \ldots, \Phi_n)
\end{equation}
\begin{equation} \label{eq:extended}
\begin{split}
\Phi_i(\mathbf{E}(t)) =& \, A_i^e(t) \mathbb{V}_i^e (\mathbf{E}(t) | \mathbf{O}_i^{e \prime}(t) ) - L_i^e(t) \\
&+ \sum_{j=1}^n A_{ij}(t) \mathbb{V}_{ij}(\mathbf{E}(t) | \mathbf{O}'_i(t) ) - \sum_{j=1}^n L_{ij}(t) \, ,
\end{split}
\end{equation}
\end{subequations}
the set of equations \eqref{eq:balance_sheet} can be rewritten in compact form:
\begin{equation} \label{eq:compact}
\mathbf{E}(t) = \Phi(\mathbf{E}(t)) \, .
\end{equation}

Therefore, performing the valuation reduces to solving the fixed-point equation for $\mathbf{E}(t)$.
The valuation functions computed at the fixed point can be interpreted as network-adjusted discount factors. 
The usual notion of discount factor captures the fact that the present value of an asset is different from its future value. 
Valuation functions account also for both the direct counterparty risk and for the network effects, which are fully incorporated in the valuations at the fixed point.
In order to implement a consistent network-adjusted valuation of interbank claims it is essential to prove the existence of solutions of (\ref{eq:compact}).

\section{Main results} \label{sec:main_results}
\noindent We now outline the most general results, which apply to generic valuation functions. Proofs are reported in the Appendix.

\begin{theorem}[Existence of greatest and least solution] \label{Tarski}
If all valuations functions in the map $\Phi$ take values in $[0,1]$ and are non-decreasing, the set of equations \eqref{eq:compact} admits a greatest solution $\mathbf{E}^{\mathrm{max}}(t)$ and a least solution $\mathbf{E}^{\mathrm{min}}(t)$.
\end{theorem}
\noindent The result above implies that the set of solutions is non-empty and that for any solution $\mathbf{E}^*(t)$, $E^{\mathrm{min}}_i(t)\leq E^*_i(t)\leq E^{\mathrm{max}}_i(t)$, for all $i$.
Within the set of solutions, the greatest solution is the most desirable outcome for all banks, as it simultaneously minimizes individual and total losses. 
In contrast, the least solution corresponds to the worst case scenario, as it simultaneously maximizes individual and total losses.
Let us explicitly note that every solution $\mathbf{E}^*(t)$ of \eqref{eq:compact} corresponds to a fixed point of the iterative map
\begin{equation} \label{eq:iter_map}
\mathbf{E}^{(k+1)}(t) = \Phi(\mathbf{E}^{(k)}(t)) \, ,
\end{equation}
and viceversa. Eq.\ \eqref{eq:iter_map} defines the Picard iteration algorithm and in principle provides a method to compute the solutions with arbitrary precision, as we will show in the following. 

Iterating the map starting from an arbitrary $\mathbf{E}^{(0)}(t)$ does not guarantee that the solutions $\mathbf{E}^{\mathrm{max}}(t)$ or $\mathbf{E}^{\mathrm{min}}(t)$ can be attained. In fact different solutions of \eqref{eq:compact} can be found depending on the chosen starting point. Moreover, some solutions might be \emph{unstable}, in the sense that, while still satisfying \eqref{eq:compact}, choosing a starting point for the Picard iteration algorithm arbitrary close to (but not equal to) such solutions, may result in the iterative map converging to another solution of \eqref{eq:compact}.
The problem of finding the least and greatest solution this problem is solved by the following two theorems.

\begin{theorem}[Convergence to the greatest solution] \label{th:convergence_greatest}
If all valuation functions in the map $\Phi$ are feasible and if $\mathbf{E}^{(0)}(t) = \mathbf{M}(t)$, then:
\begin{enumerate}
\item the sequence $\{ \mathbf{E}^{(k)}(t) \}$ is monotonic non-increasing: $\forall k \geq 0$, $\mathbf{E}^{(k+1)}(t) \leq \mathbf{E}^{(k)}(t)$,
\item the sequence $\{ \mathbf{E}^{(k)}(t) \}$ is convergent: $\lim_{k\to \infty} \mathbf{E}^{(k)}(t) = \mathbf{E}^{\infty}(t)$,
\item $\mathbf{E}^{\infty}(t)$ is a solution of \eqref{eq:compact} and furthermore $\mathbf{E}^{\infty}(t) = \mathbf{E}^{\mathrm{max}}(t)$.
\end{enumerate}
\end{theorem}
\noindent Theorem \ref{th:convergence_greatest} shows that, if the starting point of the iteration is $\mathbf{E}^{(0)}(t) = \mathbf{M}(t)$, which corresponds to taking all assets at their book value, the iterative map \eqref{eq:iter_map} converges to the greatest solution $\mathbf{E}^{\mathrm{max}}(t)$. Theorem \ref{th:convergence_greatest} guarantees that for all $\epsilon > 0$, there exists $K(\epsilon)$ such that for all $k > K(\epsilon)$ we have that $|| \mathbf{E}^{(k)}(t) - \mathbf{E}^{\mathrm{max}}(t) || < \epsilon$. In other words, once a precision $\epsilon$ has been chosen, starting from the book values of equities $\mathbf{M}(t)$, and after a finite number of iterations the Picard algorithm provides equities \eqref{eq:iter_map} that are indistinguishable from the greatest solution, within precision $\epsilon$. However, $K(\epsilon)$ is not known a priori, and at every time step it has to be checked whether the desired precision $\epsilon$ has been attained.

Mutatis mutandis, it is possible to prove that:

\begin{theorem}[Convergence to the least solution] \label{th:convergence_least}
If all valuations functions in the map $\Phi$ take values in $[0,1]$, are non-decreasing, and continuous from below, and if $\mathbf{E}^{(0)}(t) = \mathbf{m}(t)$, then:
\begin{enumerate}
\item the sequence $\{ \mathbf{E}^{(k)}(t) \}$ is monotonic non-decreasing: $\forall k \geq 0$, $\mathbf{E}^{(k+1)}(t) \geq \mathbf{E}^{(k)}(t)$,
\item the sequence $\{ \mathbf{E}^{(k)}(t) \}$ is convergent: $\lim_{k\to \infty} \mathbf{E}^{(k)}(t) = \mathbf{E}^{\infty}(t)$,
\item $\mathbf{E}^{\infty}(t)$ is a solution of \eqref{eq:compact} and furthermore $\mathbf{E}^{\infty}(t) = \mathbf{E}^{\mathrm{min}}(t)$.
\end{enumerate}
\end{theorem}
\noindent Analogous considerations to the ones proved after Theorem \ref{th:convergence_greatest} also hold in this case, implying that Theorem \ref{th:convergence_least} provides an intuitive way to compute equities in the worst case scenario.
Taken together, Theorems \ref{th:convergence_greatest} and \ref{th:convergence_least}, provide a simple algorithmic way to check whether the solution of \eqref{eq:compact} is unique within numerical precision when valuation functions are continuous (both from above and below).

Let us now put these results in the context of the existing literature. In order to prove the existence of a solution, \cite{suzuki2002valuing} and \cite{fischer2014no-arbitrage} exploit the Brouwer-Schauder fixed-point theorem, which requires payments made by each firm to be a \emph{continuous} function of the payments made by all firms. The assumption of continuity does not allow to account for default costs. However, in \cite{suzuki2002valuing} and \cite{fischer2014no-arbitrage}, the iterative map is not required to be monotonic, allowing to model some derivatives having a specific functional form. Since the Brouwer-Schauder fixed-point theorem does not give any information about the structure of the solution space (e.g.\ the existence of a greatest and a least solution) it is important to have a unique solution.\footnote{Nevertheless, in \cite{schuldenzucker2019default} it is shown that in the presence of credit-default swaps there can be multiple solutions or no solution at all.} In order to prove uniqueness \cite{suzuki2002valuing} and \cite{fischer2014no-arbitrage} resort to the additional hypothesis that the ownership matrix (the analogous of our matrix of interbank assets) is strictly left substochastic, meaning that for \emph{any} given level of seniority of the cross-holdings of debt \emph{each} firm must also hold external liabilities with the same seniority. Following \cite{eisenberg2001systemic} (see also \cite{glasserman2016contagion}), we make use of the Knaster-Tarski fixed-point theorem instead, which requires valuation functions to be monotonic -- preventing a straightforward modeling of derivatives -- and not necessarily continuous. As a consequence, default costs and analogous mechanisms can be easily accommodated in our framework (see section \ref{sec:relation}). Through the Knaster-Tarski fixed-point theorem we prove, not only the existence of a solution, but also the existence of a greatest and a least solution. Remarkably, Theorem \ref{th:convergence_greatest} shows that the greatest solution is attained if the starting point of the valuation is the book value of claims, providing a clear prescription to perform the valuation even when multiple solutions exist.

From a practical perspective it may be relevant to understand when an algorithm can terminate in a finite number of iterations.
Such results are normally proved ad hoc for a given model, meaning that they make assumptions on the explicit functional form of the valuation function \citep{hain2015valuation}. 
In contrast, here we show that such result holds for a specific topology of the network of interbank liabilities, namely a DAG (Directed Acyclic Graph), regardless of the functional form of the interbank valuation functions.
A concrete application for which this result is especially relevant is the compression of exposures \citep{derrico2018compressing}.
In fact, although perfect compression is rarely achieved in practice, techniques for compressing exposures would ideally convert any interbank network into a DAG.
\begin{proposition}[DAG] \label{pr:dag}
If the matrix defined by interbank assets $A_{ij}(t)$ is the adjacency matrix of a DAG and $\mathbb{V}_i^e(E_i(t)) = 1$, $\forall i$:
\begin{enumerate}
\item the map \eqref{eq:iter_map} converges in a finite number of iterations,
\item the solution of \eqref{eq:compact} is unique.
\end{enumerate}
\end{proposition}

\section{Relation with existing contagion models} \label{sec:relation}
\noindent We now highlight the generality of the framework outlined in section \ref{sec:framework} by presenting a few relevant examples. More specifically, we show that four different models well known in the literature can be recovered as particular cases. 
In the following we denote with $\mathbbm{1}_{x>0}$ the indicator function relative to the set defined by the condition $x>0$ and we denote with $(x)^+$ the positive part of $x$, i.e.\ $(x)^+ = (x + |x|) / 2$.

\begin{proposition}[Eisenberg and Noe] \label{pr:en}
If $t = T$ and:
\begin{enumerate}
\item $\mathbb{V}_i^e(E_i(T)) = 1$, $\forall i$,
\item $\mathbb{V}_{ij}(E_j(T)) = \mathbbm{1}_{E_j(T) \geq 0} + \left( \frac{E_j(T) + \bar{p}_j(T)}{\bar{p}_j(T)} \right)^+ \mathbbm{1}_{E_j(T)<0}$, $\forall i$, $j$
\end{enumerate}
where $\bar{p}_j(T)= L_j^e(T) + \sum_{i}L_{ji}(T)$, there is a one-to-one correspondence between the solutions of \eqref{eq:compact} and the solutions of the map $\Phi$ introduced in \cite{eisenberg2001systemic}.
\end{proposition}

From the proof in Appendix, \eqref{eq:en} implies that bank $i$ is receiving a fraction $L_{ij}(T)/ (L_j^e(T) + \sum_{i}L_{ji}(T))$ of bank $j$'s total payments, meaning that external and interbank liabilities have the same seniority.
From \eqref{eq:en_etilde} we can see that the equity is what is left after both external and interbank liabilities have been paid, and it is therefore less senior than both of them.

\begin{proposition}[Rogers and Veraart] \label{pr:rv}
If $t = T$ and:
\begin{enumerate}
\item $\mathbb{V}_i^e(E_i(T)) = 1$, $\forall i$,
\item $\mathbb{V}_{ij}(E_j(T)) = \mathbbm{1}_{E_j(T) \geq 0} + \left[(\alpha - \beta) \frac{A_j^e(T)}{\bar{p}_j(T)}  + \beta \left( \frac{E_j(T) + \bar{p}_j(T)}{\bar{p}_j(T)} \right)^+ \right] \mathbbm{1}_{E_j(T)<0}$, $\forall i$, $j$
\end{enumerate}
where $\bar{p}_j(T)= L_j^e(T) + \sum_{i}L_{ji}(T)$, there is a one-to-one correspondence between the solutions of \eqref{eq:compact} and the solutions of the map $\Phi$ introduced in \cite{rogers2013failure}.
\end{proposition}

Let us note that, unless $\alpha = \beta = 1$, $\mathbb{V}_{ij}$ is not a continuous function. 
When a bank defaults, there are two contributions to the payments to its counterparties.
First, its external assets discounted by a factor $\alpha - \beta$.
Second, its total assets $\left(E_i(\mathbf{p}^{*}(T)) + \bar{p}_i(T) \right)^+$ discounted by a factor $\beta$.
Given that total assets are the sum of interbank and external assets, putting the two terms together amounts to discounting external assets by a factor $\alpha$ and interbank assets by a factor $\beta$.
This means that, when a bank defaults, its external (interbank) assets will suddenly experience a relative loss of $\alpha - 1$ ($\beta - 1$), due e.g.\ to the necessity to liquidate them in a fire sale.

\begin{proposition}[Furfine] \label{pr:furfine}
If:
\begin{enumerate}
\item $\mathbb{V}_i^e(E_i(t)) = 1$, $\forall i$,
\item $\mathbb{V}_{ij}(E_j(t)) = \mathbbm{1}_{E_j(t) \geq 0} + R\mathbbm{1}_{E_j(t) < 0}$, $\forall i$, $j$,
\end{enumerate}
there is a one-to-one correspondence between the solutions of \eqref{eq:compact} and the solutions of the map $\Phi$ introduced in \cite{furfine2003quantifying}.
\end{proposition}

\begin{proposition}[Linear DebtRank] \label{pr:dr}
If:
\begin{enumerate}
\item $\mathbb{V}_i^e(E_i(t)) = 1$, $\forall i$,
\item $\mathbb{V}_{ij}(E_j(t)) = \min \left[ \frac{E_j^+(t)}{M_j(t)}, 1 \right]$, $\forall i$, $j$,
\end{enumerate}
there is a one-to-one correspondence between the greatest solution of \eqref{eq:compact} and the solution of the recursive map (linear version of DebtRank) introduced in \cite{bardoscia2015debtrank}.
\end{proposition}

In figure \ref{fig:ibeval} we plot several interbank valuation functions: EN (see Proposition \ref{pr:en}), Furfine (see Proposition \ref{pr:furfine}), Linear DebtRank (see Proposition \ref{pr:dr}), and ex-ante EN, which will be introduced in section \ref{sec:ex-ante}.

\begin{figure}
\centering
\includegraphics[width=0.6\columnwidth]{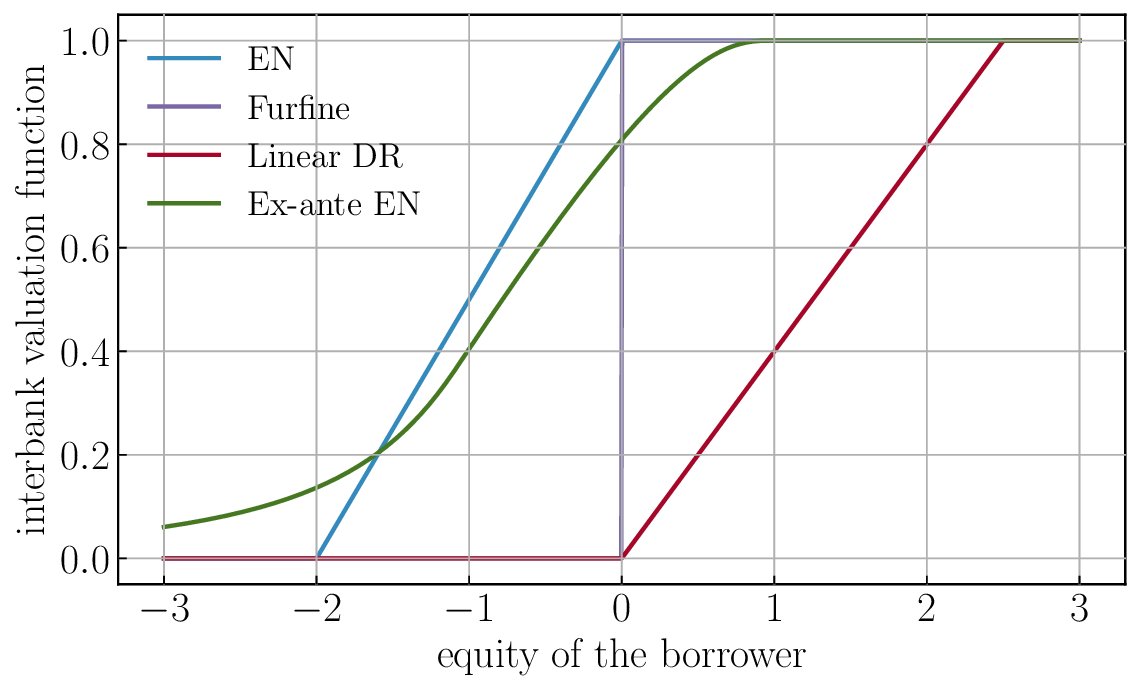}
\caption{Interbank valuation functions as a function of the equity of the borrower. Parameters as follows. EN: $\bar{p} = 2$, Furfine: $R = 1$, Linear DebtRank: $M = 2.5$, Ex-ante EN: $A^e = 1$, $\bar{p} = 2$, $\sigma = 1$.}
\label{fig:ibeval}
\end{figure}

\section{From clearing to ex-ante valuation} \label{sec:ex-ante}
\noindent On the one hand, as already remarked, clearing models allow to compute the payments that banks have to make to their counterparties at maturity. 
On the other hand, the valuation of corporate debt before maturity done by means of standard credit structural models allows each creditor only to account for their direct debtors, ignoring the indirect effect that the debtors of their debtors might have when creditors and debtors form a complex interconnected network. 
The aim of this section is to illustrate how the framework introduced in section \ref{sec:framework} can be used to bridge this gap.

In a nutshell, we take a clearing model and we show that computing expected values over the ex-ante uncertainty yields a proper ex-ante valuation model, in the sense specified in section \ref{sec:framework}.
Here we use the EN model as the starting clearing model, but in principle a different clearing model could be used as well.
As anticipated in section \ref{sec:framework}, external assets follow stochastic processes, and therefore banks face an ex-ante uncertainty on their value at maturity.
For simplicity and to make the notation clearer, we will also assume that the risk-free rate is constant and equal to zero.

The starting point is to perform a valuation of equities at time $t < T$ as in any other credit structural model, i.e.\ by taking the expected value of equities at $T$ over the unique martingale measure (EMM)\footnote{The existence of a unique EMM descends from assuming no arbitrage opportunities and a complete market. Those are standard assumptions in credit structural models. 
To the extent that one identifies expectations over the unique EMM with market values, one can say that market values computed in this way incorporate the information about the likelihood that counterparties default between $t$ and $T$.
We also note that in our framework equities are the difference between assets $A$ and liabilities $L$ and can therefore be negative. On the other hand, market prices of equities should take into account the limited liability of equity holders, who cannot go into negative equity. However, at maturity the market price of the difference between assets and liabilities can be obtained by summing the market prices of equities $\max(A - L, 0)$ to the market price of debt $\min(L, A)$ and by subtracting the book value of liabilities $L$. The argument is valid also before maturity, simply by taking the conditional expectations of equity and debt prices over the risk-neutral measure.} $\mathbb{Q}$, conditional on the filtration at time $t$:
\begin{equation} \label{eq:exp_q}
\mathbf{E}(t) = \mathbb{E}_{\mathbb{Q}}[\mathbf{E}(T) | \mathcal{F}(t)] \, .
\end{equation}
One possible method to compute the expected value on the the right-hand side of \eqref{eq:exp_q} is to perform a Monte Carlo simulation.
This is the approach proposed in \cite{fischer2014no-arbitrage}, and a variation of it has been used in \cite{elsinger2005using,elsinger2006risk}.
This requires: (i) to simulate a large number of trajectories of the stochastic processes associated with external assets up to maturity, (ii) for each simulated trajectory, to compute the solution of clearing equations with the simulated values of external assets at maturity, and (iii) to approximate the expected value with the sample average of the equities that solve the clearing equations.
However, the possibility of performing those steps relies on the implicit assumption that there is one agent who has complete knowledge of interbank assets.
The reason is that the solution of the clearing equations must be computed for a large number of potential realisations of external assets at maturity, which are unknown to individual banks at time $t$. 
Indeed, an agent with complete information is able to simulate the values of external assets at maturity \emph{and}, because she has full knowledge of interbank assets, is also able to solve all the corresponding clearing equations.
In contrast, as explained in section \ref{sec:framework}, in our approach every bank has knowledge only of their own interbank assets. 

We start by plugging the valuation functions from Proposition \ref{pr:en} into \eqref{eq:exp_q} for bank $i$:
\begin{equation} \label{eq:neva_av}
\begin{split}
E_i(t) =& \, \mathbb{E}_{\mathbb{Q}}[A_i^e(T) | \mathcal{F}(t)] - \mathbb{E}_{\mathbb{Q}}[L_i^e(T) | \mathcal{F}(t)] \\
&+ \sum_{j=1}^n \mathbb{E}_{\mathbb{Q}}[ A_{ij}(T) \mathbb{V}_{ij}^{(\mathrm{EN})}(E_j(T)) | \mathcal{F}(t)] -  \sum_{j=1}^n \mathbb{E}_{\mathbb{Q}}[ L_{ij}(T) | \mathcal{F}(t)] \, .
\end{split}
\end{equation}
Because liabilities are non-stochastic, we have that $\mathbb{E}_{\mathbb{Q}}[L_i^e(T) | \mathcal{F}(t)] = L_i^e(T)$ and that $\mathbb{E}_{\mathbb{Q}}[ L_{ij}(T) | \mathcal{F}(t)] = L_{ij}(T)$. 
As the risk-free rate is equal to zero, book values of liabilities do not need to be discounted at time $t$, meaning that $L_i^e(T) = L_i^e(t)$ and that $L_{ij}(T) = L_{ij}(t)$.
Both considerations also apply to interbank assets, implying that $\mathbb{E}_{\mathbb{Q}}[ A_{ij}(T) \mathbb{V}_{ij}^{(\mathrm{EN})}(E_j(T)) | \mathcal{F}(t)] = A_{ij}(t) \mathbb{E}_{\mathbb{Q}}[ \mathbb{V}_{ij}^{(\mathrm{EN})}(E_j(T)) | \mathcal{F}(t)]$.
As regards external assets, by definition of EMM we have that $\mathbb{E}_{\mathbb{Q}}[A_i^e(T) | \mathcal{F}(t)] = A_i^e(t)$. 
Eq.\ \eqref{eq:neva_av} now reads:
\begin{equation} \label{eq:neva_av_2}
E_i(t) = A_i^e(t) - L_i^e(t) + \sum_{j=1}^n A_{ij}(t) \mathbb{E}_{\mathbb{Q}}[ \mathbb{V}_{ij}^{(\mathrm{EN})}(E_j(T)) | \mathcal{F}(t)] - \sum_{j=1}^n L_{ij}(t) \, .
\end{equation}

We are now left with the task of computing the third term on the right-hand side of \eqref{eq:neva_av_2}, i.e.\ the expected value of $\mathbb{V}_{ij}^{(\mathrm{EN})}$:
\begin{equation} \label{eq:exante_val}
\mathbb{E}_{\mathbb{Q}}[ \mathbb{V}_{ij}^{(\mathrm{EN})}(E_j(T)) | \mathcal{F}(t)] = \mathbb{E}_{\mathbb{Q}}\left[\mathbbm{1}_{E_{j}(T)\geq0} + \left(\frac{E_{j}(T) + \bar{p}_j(T)}{\bar{p}_j(T)}\right)^+\mathbbm{1}_{E_{j}(T)<0} \Big| \mathcal{F}(t) \right] \, .
\end{equation}
Bank $i$ has to compute the expected value in \eqref{eq:exante_val} at time $t$, when the valuation is performed.
In general, we note that bank $j$ might have debtors itself.
This means that, for a fixed realisation of the stochastic processes on external assets, $E_j(T)$ might depend on the values of $j$'s interbank assets, which are not known to $i$.
Furthermore, $j$'s debtors might have debtors themselves (and so on), implying that $E_j(T)$ might in principle depend on all interbank assets.
As a consequence, bank $i$ will necessarily need to make an approximation when computing the expected value in \eqref{eq:exante_val}.
Here we assume that banks, because they have no knowledge of the debtors of their debtors, attribute the variation in the equities of their debtors to the variation in their external assets, i.e.\ $\mathbf{E}(T) \approx \mathbf{E}(t) + \mathbf{A}^e(T) - \mathbf{A}^e(t)$.
This means that the expected value in \eqref{eq:exante_val} becomes the expected value over the distributions of $A_j(T)$, the external assets at maturity, conditional on $A_j^e(t)$ the (observed) external assets at time $t$.
Moreover, after the expected value has been computed, the right-hand side of \eqref{eq:exante_val} will be an explicit function of $E_j(t)$.
Hence, \eqref{eq:neva_av_2} will have the same structure of \eqref{eq:balance_sheet}, i.e.\ provided that $\mathbb{E}_{\mathbb{Q}}[ \mathbb{V}_{ij}^{(\mathrm{EN})}(E_j(T)) | \mathcal{F}(t)]$ are feasible valuation functions, we will have an ex-ante valuation model.
As such, all results in section \ref{sec:main_results} will apply.
Theorems \ref{Tarski} and \ref{th:convergence_greatest} will ensure that there exists a greatest solution (therefore optimal for all banks) and that such solution can be computed with arbitrary precision using the Picard iteration algorithm \eqref{eq:iter_map}.
Indeed, the right-hand side of \eqref{eq:exante_val} is the expected value of a valuation function, which takes values between zero and one, and therefore it will also be between zero and one.
Similarly, the right-hand side of \eqref{eq:exante_val} is the expected value of a non-decreasing function of $E_j(T)$, but since $E_j(T) \approx E_j(t) + A^e_j(T) - A^e_j(t)$, it is also the expected value of a non-decreasing function of $E_j(t)$.  
Hence, the right-hand side of \eqref{eq:exante_val} will itself be a non-decreasing function of $E_j(t)$.
The continuity properties will in general depend on the measure $\mathbb{Q}$.
A specific example in which the right-hand side of \eqref{eq:exante_val} is continuous, and hence $\mathbb{E}_{\mathbb{Q}}[ \mathbb{V}_{ij}^{(\mathrm{EN})}(E_j(T)) | \mathcal{F}(t)]$ is a valuation function, will be discussed shortly.

Before we perform the explicit calculation of such valuation functions a few observations are in order. 
We stress that the approximation does not imply that $j$'s debtors (or other banks) do not have any impact on $i$. 
In fact, the impact of $j$'s debtors is accounted for in the equation for $E_j(t)$, which will feed into the equation for $E_i(t)$ when the fixed point of \eqref{eq:exante_val} is computed.
In practice, when \eqref{eq:exante_val} is solved iteratively, the first step of the algorithm will incorporate the effect of direct debtors into equities.
The second step will incorporate the effect of debtors of debtors, and so on.

Yet another way to interpret the approximation is to imagine that banks' valuations of interbank assets are individually risk-neutral. 
In fact, if $A_{ij}$ were the only non-zero interbank asset, the approximation would be exact because $E_j(T)$ would depend only on $A_j^e(T)$.
In this sense, bank $i$ would be computing the risk-neutral value of interbank assets, as if no other credit contract existed.

By defining $\Delta \mathbf{A}^e \equiv \mathbf{A}^e(T) - \mathbf{A}^e(t)$ and by introducing the following shorthands:
\begin{subequations} \label{eq:pd_rho}
\begin{equation} \label{eq:pd}
\begin{split}
p_j^D(E_j(t)) &= 1 - \mathbb{E}_{\mathbb{Q}}\left[\mathbbm{1}_{E_{j}(T) \geq 0} | \mathcal{F}(t) \right] \\
&= \mathbb{E}_{\mathbb{Q}}\left[\mathbbm{1}_{E_{j}(T)<0} | \mathcal{F}(t) \right] \\
&\simeq \mathbb{E}_{\mathbb{Q}}\left[\mathbbm{1}_{\Delta A^e_j < -E_j(t)} | A_j^e(t) \right]
\end{split}
\end{equation}
and\footnote{Since the book value of interbank liabilities does not change, we have that $\bar{p}_j(T) = \bar{p}_j(t)$.}:
\begin{equation} \label{eq:rho}
\begin{split}
\rho_j(E_j(t)) &= \mathbb{E}_{\mathbb{Q}}\left[\left(\frac{E_{j}(T) + \bar{p}_j(T)}{\bar{p}_j(T)}\right)^+\mathbbm{1}_{E_{j}(T)<0} \Big | \mathcal{F}(t) \right] \\
&\simeq\mathbb{E}_{\mathbb{Q}}\left[\left(\frac{E_{j}(t) + \Delta A^e_j + \bar{p}_j(t)}{\bar{p}_j(t)}\right)  \mathbbm{1}_{-\bar{p}_j(t) -E_j(t) \leq \Delta A^e_j < -E_j(t)} | A_j^e(t) \right] \, ,
\end{split}
\end{equation}
\end{subequations}
we can re-write \eqref{eq:exante_val} in the more compact form:
\begin{equation} \label{eq:genDRVal}
\mathbb{E}_{\mathbb{Q}}[ \mathbb{V}_{ij}^{(\mathrm{EN})}(E_j(T)) | \mathcal{F}(t)] \simeq 1 - p_j^D(E_j(t)) + \rho_j(E_j(t)) \, .
\end{equation}
From the second line of \eqref{eq:pd} we can see that $p_j^D(E_j(t))$ is, by definition, the probability that bank $j$ defaults at maturity. 
Analogously, from \eqref{eq:rho} we can see that $\rho_j(E_j(t))$ is an endogenous recovery rate.
In fact, when $E_{j}(T)<0$, $E_{j}(T) + \bar{p}_j(T)$ is equal to $j$'s total assets, which are smaller than its liabilities. 
Hence, $\rho_j(E_j(t))$ is equal to the (conditional) expected value of $j$'s assets at maturity when $j$ defaults divided by its total liabilities, i.e.\ the (conditional) expected value of the fraction of interbank assets that a creditor can expect to recover.
From \eqref{eq:genDRVal} we can see that the valuation function can be thought of as the expectation over a two-valued probability distribution: if the debtor $j$ does not default at maturity, $i$ will recover its interbank asset in full, while if $j$ defaults, $i$ will recover the endogenous recovery rate.
Thus, \eqref{eq:genDRVal} can be interpreted as a generalisation to endogenous recovery rates of the valuation mechanisms in \cite{bardoscia2015distress} (see (7) therein) and in \cite{bardoscia2017pathways} (see (2) in Supplementary Methods), in which the recovery rate is exogenous.

Finally, as we approach maturity, the valuation functions in \eqref{eq:genDRVal} approach the EN valuation functions in Proposition \ref{pr:en}.

\begin{proposition} \label{pr:exante_to_en}
In the limit in which the maturity is approached, i.e.\ $t \to T$, the interbank valuation function \eqref{eq:genDRVal} converges to the interbank valuation function of EN (Proposition \ref{pr:en}).
\end{proposition}

\subsection{Ex-ante valuation with geometric Brownian motion} \label{subsec:bm}
\noindent We now explicitly compute the probability of default and the endogenous recovery rate in \eqref{eq:pd_rho} assuming that external assets follow independent geometric Brownian motions:
\begin{equation} \label{eq:gbm}
\mathrm{d} A_i^e(s) = \mu_i A_i^e(s) \mathrm{d} s +  \sigma_i A_i^e(s) \mathrm{d} W_i(s) \qquad \forall s \in [t, T], \, i.
\end{equation}
The probability density function of $\Delta A_i^e$ in the measure $\mathbb{Q}$ is:
\begin{equation} \label{eq:logndist}
p(\Delta A_i^e) = \frac{1}{\sqrt{2\pi (T-t)}\sigma_i (\Delta A_i^e + A_i^e(t))}e^{\frac{-\left[\log\left(1+\frac{\Delta A_i^e}{A_i^e(t)}\right)+\frac{1}{2}\sigma_i^2(T-t)\right]^2}{2\sigma_i^2(T-t)}} \, .
\end{equation}
From \eqref{eq:pd_rho} we then have:
\begin{subequations} \label{eq:pd_rho_gbm}
\begin{equation}\label{eq:pd_gbm}
p_{j}^D(E_j(t)) = \frac{1}{2} \left[ 1 + \mathrm{erf} \left[\frac{\log(1-E_j(t)/A_j^e(t)) + \sigma_j^2 (T-t)/2}{\sqrt{2(T-t)} \sigma_j} \right] \right] \mathbbm{1}_{E_j(t)<A_j^{e}(t)}
\end{equation}
\begin{equation}\label{eq:rho_gbm}
\rho_{j}(E_j(t)) = \left( 1+\frac{E_j(t)}{\bar{p}_j(t)} \right) \left[ p_j^D(E_j(t)) - p_j^D(E_j(t) + \bar{p}_j(t)) \right] + \frac{A_j^{e}(t)}{2 \bar{p}_j}c_{j}(E_j(t))
\end{equation}
with
\begin{align*}
c_{j}(E_j(t)) =& -\text{erf}\left[{\frac{\sigma_j^2(T-t)/2 - \log\left(1-E_j(t)/A_j^e(t)\right)}{\sqrt{2(T-t)}\sigma_j}}\right] \mathbbm{1}_{E_j(t)<A_j^e(t)} \\
&-\text{erf}\left[{\frac{\sigma_j^2(T-t)/2 + \log\left(1-E_j(t)/A_j^e(t)\right)}{\sqrt{2(T-t)}\sigma_j}}\right] \mathbbm{1}_{E_j(t)<A_j^e(t)} \\
&+\text{erf}\left[{\frac{\sigma_j^2(T-t)/2 + \log\left(1-(E_j(t)+\bar{p}_j(t))/A_j^e(t)\right)}{\sqrt{2(T-t)}\sigma_j}}\right] \mathbbm{1}_{E_j(t)<A_j^e(t)-\bar{p}_j(t)} \\
&+\text{erf}\left[{\frac{\sigma_j^2(T-t)/2 - \log\left(1-(E_j(t)+\bar{p}_j(t))/A_j^e(t)\right)}{\sqrt{2(T-t)}\sigma_j}}\right] \mathbbm{1}_{E_j(t)<A_j^e(t)-\bar{p}_j(t)} \, .
\end{align*}
\end{subequations}
By plugging \eqref{eq:pd_rho_gbm} into \eqref{eq:genDRVal} it is easy to show that $\mathbb{E}_{\mathbb{Q}}[ \mathbb{V}_{ij}^{(\mathrm{EN})}(E_j(T)) | \mathcal{F}(t)]$ is actually a continuous function of $E_j(t)$ (both from above and from below), and therefore a feasible valuation function.

\subsection{Stress testing: Merton vs network valuation }
\noindent As a proof of concept, here we carry out a stress test on a small financial system composed by three banks, $A$, $B$, $C$. We choose a simple ring topology, $A \to B \to C \to A$ with the following parameters:
\begin{equation}
\mathbf{A}^e(t) = \begin{pmatrix}
       10 \\
       4  \\
       1.5
     \end{pmatrix}
\qquad
\mathbf{L}^e(t) = \begin{pmatrix}
       9  \\
       3  \\
       0.5
     \end{pmatrix}
\qquad
A(t) = \begin{pmatrix}
        0 & 0.8 & 0 \\
        0 & 0 & 0.8 \\
        0.8 & 0 & 0
       \end{pmatrix} \, ,
\end{equation}
so that all three banks have a book value of their equity equal to one. Total leverages, defined as the ratio between total assets and book values of equity, range from 10.8 to 2.3. Our stress test consists of applying an exogenous shock to the external assets of all banks, resulting in a devaluation, in relative terms, by a factor $\alpha$, i.e.\ $A^e_i(t) \to (1-\alpha) A^e_i(t)$. The variation in external assets of bank $i$, measured as the difference between its external assets before the shock and its external assets after the shock is $\Delta A^e_i = \alpha A^e_i(t)$. Using \eqref{eq:balance_sheet} we can readily compute the corresponding variation in equity, again measured as the difference between the equity before the shock (i.e.\ its book value) and the equity after the shock: $\Delta E_i =  \alpha A^e_i(t) + \sum_j A_{ij}(t)\left[ 1 - \mathbb{V}_{ij}(E_j^*(t)) \right]$. 
The network contribution can be quantified as the total losses in the system minus the losses directly caused by the exogenous shock: $\sum_i \Delta E_i - \Delta A^e_i = \sum_{ij} A_{ij}(t)\left[1 - \mathbb{V}_{ij}(E_j^*(t)) \right]$, which can be conveniently normalised by its maximum, $\sum_{ij} A_{ij}(t)$:
\begin{equation} \label{eq:net_eff}
\frac{\sum_i \Delta E_i - \Delta A^e_i}{\sum_{ij} A_{ij}(t)} = \frac{\sum_{ij} A_{ij}(t) \left[1 - \mathbb{V}_{ij}(E_j^*(t))\right]}{\sum_{ij} A_{ij}(t)} \, .
\end{equation}
In the left panel of figure \ref{fig:illustrative} we show the behaviour of the quantity \eqref{eq:net_eff} as a function of the exogenous shock of magnitude $\alpha$ on external assets, for several valuation functions. For Furfine we take the exogenous recovery rate $R = 0$, while for ex-ante EN we take \eqref{eq:gbm} with external assets volatility $\sigma_i (T - t) = 0.5$, for all banks. 
Interestingly, we can see that the network contribution for ex-ante EN is always larger than for EN.
We point out that this behaviour is not in contrast with the behavior shown in figure \ref{fig:ibeval}, where the EN interbank valuation function becomes smaller than the ex-ante EN interbank valuation function for sufficiently small values of equities (which correspond to sufficiently large shocks).
In fact, in figure \ref{fig:ibeval} the ex-ante EN interbank valuation function is computed for fixed external assets $\mathbf{A}^e(t)$.
In contrast, in figure \ref{fig:illustrative} the ex-ante EN interbank valuation functions are computed at the post-shock values of external assets $\mathbf{A}^e(t)$, which obviously vary with the exogenous shock to external assets.

Another way to assess the extent of the network contribution is the following. Let us imagine that each bank wants to valuate the interbank assets of its counterparty using the standard Merton approach. This amounts to using the valuation function \eqref{eq:genDRVal} and evaluating it in the book value of the equity of the counterparty. Hence, the lender $i$ discounts its interbank assets $A_{ij}(t)$ towards the borrower $j$ by a factor $\mathbb{V}_{ij}(M_j(t))$. If the same valuation is performed using using ex-ante EN, the discount factor equals to $\mathbb{V}_{ij}(E_j^*(t))$. In the right panel of figure \ref{fig:illustrative} we show the difference between such discount factors, i.e.\ between the discount factor of the valuation of an interbank claim performed with the standard Merton approach and the valuation of an interbank claim performed with ex-ante EN valuation functions \eqref{eq:gbm} with $\sigma_i (T - t) = 0.5$, for all banks. 
In this example, from the right panel of figure \ref{fig:illustrative} we can see that for bank C, which holds a claim towards bank A, such difference can be larger than 60\% (when $\alpha =1$). 
Since the book value of the interbank asset held by bank $C$ is equal to 0.8, by using the Merton model we would overestimate its 
value, in absolute terms, by $0.8 \cdot 0.6 \approx 0.5$, which is about 50\% of the book value of bank $C$'s equity (which is equal to one in this example).
Furthermore, the larger the shock to external assets, the larger the difference between the two discount factors.

\begin{figure}
\centering
\includegraphics[width=\columnwidth]{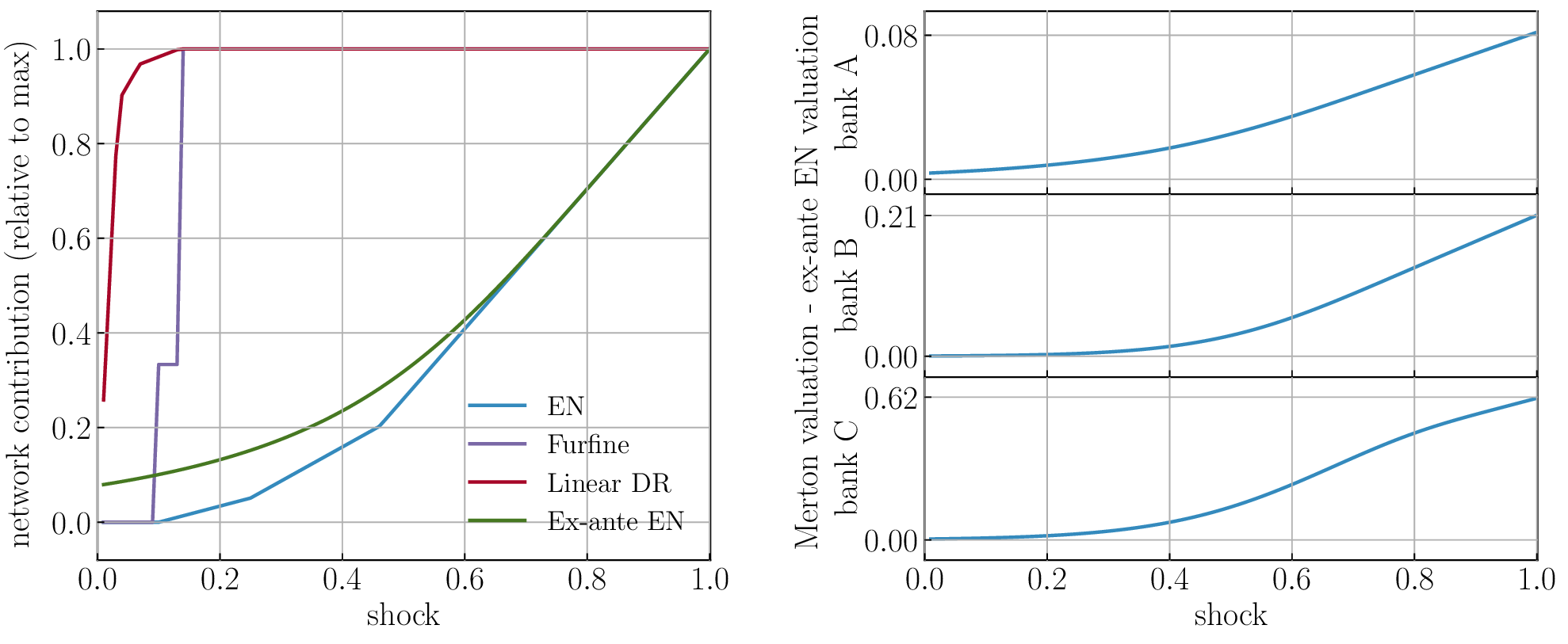}
\caption{Stress test consisting in applying an exogenous shock to external assets of all banks and by re-evaluating interbank claims. Left panel: network contribution (measured as the relative loss due to the revaluation of interbank assets) as a function of the exogenous shock, for several valuation functions. Right panel: for each of the three banks in the example network, the difference is shown between the discount factor relative to the valuation of their interbank claim performed with the standard Merton approach and the discount factor relative to the valuation of their interbank claim performed with ex-ante EN.}
\label{fig:illustrative}
\end{figure}

\section{Conclusions} \label{sec:Conclusions}
\noindent In this paper, we introduce a general framework to perform an ex-ante and network-adjusted valuation of financial institutions' interbank claims. On the one hand, our framework encompasses some of the most widely used models of financial contagion \citep{eisenberg2001systemic,furfine2003quantifying,rogers2013failure,bardoscia2015debtrank}, in the precise sense that the model is equivalent to those models for specific choices of the valuation functions and the parameters. 
On the other hand, our framework relates also to the stream of literature \citep{fischer2014no-arbitrage,suzuki2002valuing} carrying out the valuation of claims \`{a} la Merton when cross-holdings of debt exist between different firms. 

Our main result is that, under mild assumptions about valuation functions, the valuation problem admits a greatest solution, i.e.\ a solution in which the equities of all banks are maximal. 
Moreover, we provide a simple iterative algorithm to compute such solution.
Finally, we show how an ex-ante valuation model can be derived from a clearing model, i.e.\ from a model in which the valuation is performed at maturity.

A natural application of our framework is in devising \emph{stress-tests} to assess losses on banks' portfolios in a network of liabilities, conditional to shocks on their external assets in order to determine capital requirements and value at risk. Indeed, to any given shock on the external assets of the banks it corresponds a different valuation of banks' equities. Therefore, by assuming a known distribution of shocks, one can derive a corresponding distribution of equity losses. Such distribution can then be taken as the input of any axiomatic risk measure. 
Finally, one could embed our framework into a full-fledged economy with non-financial firms, households, and a government. 
Indeed, \cite{gray2010measuring} build an extension of the Merton model that, while abstracting from the network of individual firms, focuses on interconnections between sectorial balance sheets, thereby allowing to discuss the main transmission channels between them.

\section*{Data Availability Statement}
\noindent Data sharing is not applicable to this article as no new data were created or analyzed in this study. A Python implementation of the network valuation framework, including all valuation functions introduced here, is available at: \url{https://github.com/marcobardoscia/neva}.

\section*{Acknowledgments}
\noindent 
All authors contributed to the concept of the study. 
PB, MB, and FC derived the results. MDE, GV, GC, and SB reviewed the results and their significance. 
PB, MB and SB prepared the manuscript. 
PB, MB, FC, MDE, GC, and SB edited and revised the manuscript.
The first version of this paper was prepared when PB, MB, MDE, and GV were employed by the Department of Banking and Finance of the University of Zurich.
PB and MB thank the London Institute for Mathematical Sciences for material support in the early stages of this work.
PB, MB, MDE, GV, SB, and GC acknowledge funding from the European Union's FP7 and Horizon 2020 research and innovation programmes under grant agreements nr.\ 610704 (FET Project SIMPOL) and nr.\ 640772 (FET Project DOLFINS).
SB acknowledges the Swiss National Fund Professorship grant nr.\ PP00P1-144689. 
GC acknowledges the funding of the European Union's FP7 and Horizon 2020 research and innovation programmes under grant agreements nr.\ 317532 (FET IP Project MULTIPLEX), nr.\ 654024 (SoBigData), nr.\ 676547 (CoeGSS), nr.\ 952026 (HumanE-AI-Net), and nr.\ 871042 (SoBigData++). 
FC acknowledges support of the Economic and Social Research Council (ESRC) in funding the Systemic Risk Centre (ES/K002309/1).

\bibliographystyle{apalike}
\bibliography{finNetBib}

\section*{Appendix: Proofs of Theorems and Propositions}

\setcounter{equation}{0}
\renewcommand{\theequation}{A.\arabic{equation}}

\noindent For the sake of readability, in the proofs we suppress the explicit dependence on $t$ whenever there is no risk of ambiguity.

\renewcommand{\thetheorem}{\ref{Tarski}}
\begin{theorem}[Existence of greatest and least solution]
If all valuations functions in the map $\Phi$ take values in $[0,1]$ and are non-decreasing, the set of equations \eqref{eq:compact} admits a greatest solution $\mathbf{E}^{\mathrm{max}}(t)$ and a least solution $\mathbf{E}^{\mathrm{min}}(t)$.
\end{theorem}
\begin{proof}
To prove it we just need to show that: (a) the function $\Phi$ maps a complete lattice into itself, $\Phi : \Lambda \rightarrow \Lambda$, (b) the function $\Phi$ is an order-preserving function.
To prove (a) we notice that if valuation functions are feasible then:
$$
\forall \mathbf{E} \in \mathbbm{R}^n \quad  m_i = -L^{e}_i - \sum_{j}L_{ij}\leq \Phi_i(\mathbf{E}) \leq A^{e}_i-L^{e}_i + \sum_{j}A_{ij} - \sum_{j}L_{ij} = M_i
$$
and consequently $\Lambda=\varprod_{i=1}^n [m_i,M_i]$ is a complete lattice such that $\Phi : \Lambda \rightarrow \Lambda$, that proves (a). Since $\Phi$ is a linear combination of monotonic non-decreasing functions in $\mathbf{E}$, then  $\forall \mathbf{E},\mathbf{E'}$ if $\mathbf{E} \leq \mathbf{E}'$, follows $\Phi(\mathbf{E}) \leq \Phi(\mathbf{E'})$, where the partial ordering relation in $\Lambda$ is component-wise, i.e. $\mathbf{x}\leq\mathbf{y}$ iff $\forall i$ $x_i \leq y_i$. So both conditions (a) and (b) hold and the Knaster-Tarski theorem applies. The set of solutions $S$ of \eqref{eq:compact} is then a complete lattice, therefore it is non-empty (the empty set cannot contain its own supremum) and, more importantly, it admits a supremum solution, $\mathbf{E}^{\mathrm{max}}$, and an infimum solution, $\mathbf{E}^{\mathrm{min}}$, such that $\forall \mathbf{E}^{*} \in S$, $\mathbf{E}^{\mathrm{min}} \leq \mathbf{E}^{*} \leq \mathbf{E}^{\mathrm{max}}$. 
\end{proof}

\renewcommand{\thetheorem}{\ref{th:convergence_greatest}}
\begin{theorem}[Convergence to the greatest solution]
If all valuation functions in the map $\Phi$ are feasible and if $\mathbf{E}^{(0)}(t) = \mathbf{M}(t)$, then:
\begin{enumerate}
\item the sequence $\{ \mathbf{E}^{(k)}(t) \}$ is monotonic non-increasing: $\forall k \geq 0$, $\mathbf{E}^{(k+1)}(t) \leq \mathbf{E}^{(k)}(t)$,
\item the sequence $\{ \mathbf{E}^{(k)}(t) \}$ is convergent: $\lim_{k\to \infty} \mathbf{E}^{(k)}(t) = \mathbf{E}^{\infty}(t)$,
\item $\mathbf{E}^{\infty}(t)$ is a solution of \eqref{eq:compact} and furthermore $\mathbf{E}^{\infty}(t) = \mathbf{E}^{\mathrm{max}}(t)$.
\end{enumerate}
\end{theorem}
\begin{proof}
Convergence will be proved by induction. For $n = 0$ we have
$$ \mathbf{E}^{(1)} = \Phi(\mathbf{E}^{(0)}) \leq \mathbf{M} = \mathbf{E}^{(0)} $$
Assume now that the claim is true for all $0 \leq m \leq n$, then
$$ \mathbf{E}^{(n+1)} = \Phi(\mathbf{E}^{(n)}) \leq \Phi(\mathbf{E}^{(n - 1)}) = \mathbf{E}^{(n)}$$
where we have used the fact that $\Phi$ is monotonic non-decreasing and $\mathbf{E}^{(n)} \leq \mathbf{E}^{(n - 1)}$ by hypothesis,
We know that $\{\mathbf{E}^{(n)}\}$ is bounded below and monotonic non-increasing, by the Monotone Convergence Theorem we have that $\mathbf{E}^{*} = \lim_{n\rightarrow\infty} \mathbf{E}^{(n)} = \inf_n\{\mathbf{E}^{(n)}\}$ exists and is finite.
By hypothesis $\Phi$ is continuous from above because under assumptions of Theorem (\ref{th:convergence_greatest}) we know that the valuation functions are feasible, hence
$$ \Phi(\mathbf{E}^{*}) = \Phi(\lim_{n} \mathbf{E}^{(n)}) = \lim_n \Phi(\mathbf{E}^{(n)}) = \lim_n \mathbf{E}^{(n+1)} = \mathbf{E}^{*}$$
So that $\mathbf{E}^{*} \in S$.
We will now prove it must be that $\mathbf{E}^{*} = \mathbf{E}^{\mathrm{max}}$.
First we need to establish a preliminary result, namely that $\mathbf{E}^{(n)} \geq \mathbf{E}^{\mathrm{max}}, \forall n$.
Reasoning by induction, it is trivially true for the initial point that $\mathbf{E}^{(0)}\geq \mathbf{E}^{\mathrm{max}}$.
Suppose now that it is true up to a given $\bar{n}$, $\mathbf{E}^{(\bar{n})} \geq \mathbf{E}^{\mathrm{max}}$ then, since $\Phi$ is order-preserving,
$$ \mathbf{E}^{(\bar{n}+1)} = \Phi(\mathbf{E}^{(\bar{n})}) \geq \Phi(\mathbf{E}^{\mathrm{max}}) = \mathbf{E}^{\mathrm{max}}$$
Now, knowing that $\mathbf{E}^{(n)} \geq \mathbf{E}^{\mathrm{max}}, \forall n$ we have that $\mathbf{E}^{*} = \inf_n\{\mathbf{E}^{(n)}\} \geq \mathbf{E}^{\mathrm{max}}$. But $\mathbf{E}^{*} \in S$, hence $\mathbf{E}^{*} = \mathbf{E}^{\mathrm{max}}$. 
\end{proof}

\renewcommand{\thetheorem}{\ref{th:convergence_least}}
\begin{theorem}[Convergence to the least solution]
If all valuations functions in the map $\Phi$ take values in $[0,1]$, are non-decreasing, and continuous from below, and if $\mathbf{E}^{(0)}(t) = \mathbf{m}(t)$, then:
\begin{enumerate}
\item the sequence $\{ \mathbf{E}^{(k)}(t) \}$ is monotonic non-decreasing: $\forall k \geq 0$, $\mathbf{E}^{(k+1)}(t) \geq \mathbf{E}^{(k)}(t)$,
\item the sequence $\{ \mathbf{E}^{(k)}(t) \}$ is convergent: $\lim_{k\to \infty} \mathbf{E}^{(k)}(t) = \mathbf{E}^{\infty}(t)$,
\item $\mathbf{E}^{\infty}(t)$ is a solution of \eqref{eq:compact} and furthermore $\mathbf{E}^{\infty}(t) = \mathbf{E}^{\mathrm{min}}(t)$.
\end{enumerate}
\end{theorem}
\begin{proof}
Convergence will be proved by induction. For $n = 0$ we have
$$ \mathbf{E}^{(1)} = \Phi(\mathbf{E}^{(0)}) \geq \mathbf{m} = \mathbf{E}^{(0)} $$
Assume now that the claim is true for all $0 \leq m \leq n$, then
$$ \mathbf{E}^{(n+1)} = \Phi(\mathbf{E}^{(n)}) \geq \Phi(\mathbf{E}^{(n - 1)}) = \mathbf{E}^{(n)}$$
where we have used the fact that $\Phi$ is monotonic non-decreasing and $\mathbf{E}^{(n)} \geq \mathbf{E}^{(n - 1)}$ by hypothesis.
We know that $\{\mathbf{E}^{(n)}\}$ is bounded above and monotonic non-decreasing, by the Monotone Convergence Theorem we have that $\mathbf{E}^{*} = \lim_{n} \mathbf{E}^{(n)} = \sup_n\{\mathbf{E}^{(n)}\}$ exists and is finite.
By hypothesis $\Phi$ is continuous from below, hence
$$ \Phi(\mathbf{E}^{*}) = \Phi(\lim_{n} \mathbf{E}^{(n)}) = \lim_{n\rightarrow\infty} \Phi(\mathbf{E}^{(n)}) = \lim_{n\rightarrow\infty}  \mathbf{E}^{(n+1)} = \mathbf{E}^{*}$$
So that $\mathbf{E}^{*} \in S$.
We will now prove it must be that $\mathbf{E}^{*} = \mathbf{E}^{\mathrm{min}}$.
First we need to establish a preliminary result, namely that $\mathbf{E}^{(n)} \leq \mathbf{E}^{\mathrm{min}}, \forall n$.
Reasoning by induction, it is trivially true for the initial point that $\mathbf{E}^{(0)}\leq \mathbf{E}^{\mathrm{min}}$.
Suppose now that it is true up to a given $\bar{n}$, $\mathbf{E}^{(\bar{n})} \leq \mathbf{E}^{\mathrm{min}}$ then, since $\Phi$ is order-preserving,
$$ \mathbf{E}^{(\bar{n}+1)} = \Phi(\mathbf{E}^{(\bar{n})}) \leq \Phi(\mathbf{E}^{\mathrm{min}}) = \mathbf{E}^{\mathrm{min}}$$
Now, knowing that $\mathbf{E}^{(n)} \leq \mathbf{E}^{\mathrm{min}}, \forall n$ we have that $\mathbf{E}^{*} = \sup_n\{\mathbf{E}^{(n)}\} \leq \mathbf{E}^{\mathrm{min}}$. But $\mathbf{E}^{*} \in S$, hence $\mathbf{E}^{*} = \mathbf{E}^{\mathrm{min}}$. 
\end{proof}

\renewcommand{\theproposition}{\ref{pr:dag}}
\begin{proposition}[DAG]
If the matrix defined by interbank assets $A_{ij}(t)$ is the adjacency matrix of a DAG and $\mathbb{V}_i^e(E_i(t)) = 1$, $\forall i$:
\begin{enumerate}
\item the map \eqref{eq:iter_map} converges in a finite number of iterations,
\item the solution of \eqref{eq:compact} is unique.
\end{enumerate}
\end{proposition}
\begin{proof}
We define source banks as those banks that do not hold interbank assets, i.e.\ $S_0=\{i :\, A_{ij}=0, \forall j \}$, which is a non-empty set if the matrix of interbank exposures is a DAG. We then partition banks based on the maximum graph distance from the set of source banks $S_0$, the partition being $\{S_d\}_{d=0}^{d_{\mathrm{max}}}$. Starting from the initial condition $\mathbf{M}$, banks in $S_0$ converge in zero iterations to their book value as their equity does not depend on the equity of any other bank (neither their own). Banks in $S_1$ converge in one iteration as their equity only depends on the equities of banks in $S_0$. By induction, banks in $S_{d_{\mathrm{max}}}$ converge in $d_{\mathrm{max}}$ iterations. Starting from the initial condition $\mathbf{m}$ banks in $S_0$ converge in one iteration to their book value as the Picard iteration algorithm corrects the value of their equities exactly in one iteration. Consequently, $\Phi^{(d_{\mathrm{max}})}(\mathbf{M})=\Phi^{(d_{\mathrm{max}}+1)}(\mathbf{m})$, and therefore all banks converge to $\mathbf{E}^{\mathrm{min}}=\mathbf{E}^{\mathrm{max}}$ in (at most) $d_{\mathrm{max}} + 1$ iterations.
\end{proof}

\renewcommand{\theproposition}{\ref{pr:en}}
\begin{proposition}[Eisenberg and Noe]
If $t = T$ and:
\begin{enumerate}
\item $\mathbb{V}_i^e(E_i(T)) = 1$, $\forall i$,
\item $\mathbb{V}_{ij}(E_j(T)) = \mathbbm{1}_{E_j(T) \geq 0} + \left( \frac{E_j(T) + \bar{p}_j(T)}{\bar{p}_j(T)} \right)^+ \mathbbm{1}_{E_j(T)<0}$, $\forall i$, $j$
\end{enumerate}
where $\bar{p}_j(T)= L_j^e(T) + \sum_{i}L_{ji}(T)$, there is a one-to-one correspondence between the solutions of \eqref{eq:compact} and the solutions of the map $\Phi$ introduced in \cite{eisenberg2001systemic}.
\end{proposition}
\begin{proof}
As already noted, in EN the valuation happens \emph{at} maturity, $t = T$. 
Under the assumptions of (i) limited liabilities, (ii) priority of debt over equity, (iii) proportional repayments, EN aims at computing a clearing payment vector $\mathbf{p}^*(T)$ whose component $p_i^*(T)$ is the total payment made by bank $i$ to its counterparties. 
To conform to their notation, we also introduce the obligation vector $\bar{\mathbf{p}}(T)$, defined as $\bar{p}_i(T) = L_i^e(T) + \sum_j L_{ij}(T)$, which is the total liabilities of bank $i$. \cite{eisenberg2001systemic} show that:
\begin{equation} \label{eq:en}
p_i^*(T) = \min \left[ e_i(T) + \sum_j L_{ji}(T) \frac{p_j^*(T)}{\bar{p}_j(T)}, \bar{p}_i(T) \right] \, ,
\end{equation}
where $e_i(T) = A_i^e(T)$. 
Eq.\ \eqref{eq:en} can be equivalently rewritten as:
\begin{subequations}
\begin{equation} \label{eq:p_star}
p_i^*(T) = \bar{p}_i(T) \mathbbm{1}_{E_i(\mathbf{p}^{*}(T)) \geq 0} + \left[ E_i(\mathbf{p}^{*}(T)) + \bar{p}_i(T) \right]^+ \mathbbm{1}_{E_i(\mathbf{p}^{*}(T))<0} \, ,
\end{equation}
with
\begin{equation} \label{eq:en_etilde}
E_i(\mathbf{p}(T)) = A_i^e(T) - L_i^e(T) + \sum_{j} A_{ij}(T) \frac{p_j(T)}{\bar{p}_j(T)} - \sum_j L_{ij}(T) \, .
\end{equation}
\end{subequations}
The above equations are equivalent to \eqref{eq:balance_sheet} by choosing the valuation functions as in the hypotheses of the Proposition \ref{pr:en}. 
In fact, when $E_j(T) > 0$, the cash inflow of bank $j$ is enough to cover its due payments, and therefore $\bar{\mathbf{p}}(T) = \mathbf{p}^*(T)$. In contrast, when $E_j(T) < 0$, bank $j$ employs its residual assets $\left[ E_j(T) + \bar{p}_j(T) \right]^+$ to repay its creditors proportionally as much as it can.
\end{proof}

\renewcommand{\theproposition}{\ref{pr:rv}}
\begin{proposition}[Rogers and Veraart] 
If $t = T$ and:
\begin{enumerate}
\item $\mathbb{V}_i^e(E_i(T)) = 1$, $\forall i$,
\item $\mathbb{V}_{ij}(E_j(T)) = \mathbbm{1}_{E_j(T) \geq 0} + \left[(\alpha - \beta) \frac{A_j^e(T)}{\bar{p}_j(T)}  + \beta \left( \frac{E_j(T) + \bar{p}_j(T)}{\bar{p}_j(T)} \right)^+ \right] \mathbbm{1}_{E_j(T)<0}$, $\forall i$, $j$
\end{enumerate}
where $\bar{p}_j(T)= L_j^e(T) + \sum_{i}L_{ji}(T)$, there is a one-to-one correspondence between the solutions of \eqref{eq:compact} and the solutions of the map $\Phi$ introduced in \cite{rogers2013failure}.
\end{proposition}
\begin{proof}
The proof is entirely analogous to the proof of Proposition \ref{pr:en}. 
Similarly to \eqref{eq:p_star}, payments as functions of equities are given by:
\begin{equation*}
\begin{split}
p_i^*(T) =& \, \bar{p}_i(T) \mathbbm{1}_{E_i(\mathbf{p}^{*}(T)) \geq 0} \\
&+ \left[(\alpha - \beta) A_i^e(T)  + \beta \left(  E_i(\mathbf{p}^{*}(T)) + \bar{p}_i(T) \right)^+ \right] \mathbbm{1}_{E_i(\mathbf{p}^{*}(T))<0}.
\end{split}
\end{equation*}
\end{proof}

\renewcommand{\theproposition}{\ref{pr:furfine}}
\begin{proposition}[Furfine]
If:
\begin{enumerate}
\item $\mathbb{V}_i^e(E_i(t)) = 1$, $\forall i$,
\item $\mathbb{V}_{ij}(E_j(t)) = \mathbbm{1}_{E_j(t) \geq 0} + R\mathbbm{1}_{E_j(t) < 0}$, $\forall i$, $j$,
\end{enumerate}
there is a one-to-one correspondence between the solutions of \eqref{eq:compact} and the solutions of the map $\Phi$ introduced in \cite{furfine2003quantifying}.
\end{proposition}
\begin{proof}
According to the Furfine algorithm a counterparty with non-negative equity is always able to fully repay its liabilities, while, if its equity is negative, it will only repay a fraction $R$ of them. This is exactly what the valuation function in Proposition \ref{pr:furfine} accounts for. 
\end{proof}

\renewcommand{\theproposition}{\ref{pr:dr}}
\begin{proposition}[Linear DebtRank]
If:
\begin{enumerate}
\item $\mathbb{V}_i^e(E_i(t)) = 1$, $\forall i$,
\item $\mathbb{V}_{ij}(E_j(t)) = \min \left[ \frac{E_j^+(t)}{M_j(t)}, 1 \right]$, $\forall i$, $j$,
\end{enumerate}
there is a one-to-one correspondence between the greatest solution of \eqref{eq:compact} and the solution of the recursive map (linear version of DebtRank) introduced in \cite{bardoscia2015debtrank}.
\end{proposition}
\begin{proof}
The easiest way to prove the correspondence is to start from $\mathbf{M}(t)$ and to compute the incremental variation of the iterative map \eqref{eq:iter_map}, which in this case reads: $E_i^{(k+1)}(t) - E_i^{(k)}(t) = \sum_j A_{ij}(t) \frac{\left(E_j^{(k)}(t)\right)^+ - \left(E_j^{(k-1)}(t)\right)^+}{M_j(t)}$, for all $i$. Starting the Picard iteration algorithm from $\mathbf{M}(t)$ we recover (7) in \cite{bardoscia2015debtrank}, in which $\mathbf{M}(t)$ has been denoted with $\mathbf{E}(0)$. As soon as the equity of bank $j$ becomes equal to zero in the iterative map in \cite{bardoscia2015debtrank}, it will not change anymore, which is consistent with the incremental variation derived above.
\end{proof}

\renewcommand{\theproposition}{\ref{pr:exante_to_en}}
\begin{proposition}
In the limit in which the maturity is approached, i.e.\ $t \to T$, the interbank valuation function \eqref{eq:genDRVal} converges to the interbank valuation function of EN (Proposition \ref{pr:en}).
\end{proposition}
\begin{proof}
First we notice that, as $t \to T$ the variation in external assets goes to zero with probability approaching one, and therefore from \eqref{eq:pd_rho} we have that $p_j^D(E_j(t)) \to \mathbbm{1}_{E_j(T)<0}$ and that $\rho_j(E_j(t)) \to \left(\frac{E_j(T) + \bar{p}_j(T)}{\bar{p}_j(T)}\right)^+\mathbbm{1}_{E_j(T)<0}$, from which the proposition easily follows.
\end{proof}

\end{document}